\documentclass[A4paper,10pt,journal,twocolumn,twoside]{IEEEtran}

\usepackage{epsfig,amsfonts,subfigure,amsbsy,bm,mathrsfs}
\usepackage{graphicx,cite,amssymb,amsmath,color}
\usepackage{amsmath,amssymb,amscd,amsthm,mathabx}
\usepackage{mathtools}
\usepackage{multirow}
\usepackage[nolist]{acronym}
\mathtoolsset{showonlyrefs=true}

\newtheorem{prop}{Proposition}

\def\forcemath#1{\ifmmode #1 \else $#1$\fi}

\newcommand{\B}{\forcemath{\bm{B}}}
\newcommand{\proto}{\mathscr{P}}

\newcommand{\BMA}[2]{\forcemath{{\left(#1 \,\, #2 \right)}}}
\newcommand{\BMB}[6]{\forcemath{{\left(\begin{array}{c|cc}
	#1 & #3 & #5 \\
	#2 & #4 & #6
	\end{array}
	\right)}}}

\newcommand{\thrspa}{\forcemath{{\delta^\star_{\mathsf{SPA}}}}}
\newcommand{\thre}{\forcemath{{\delta^\star_{\mathsf{E}}}}}

\newcommand{\ensemble}{\forcemath{{\mathscr{C}}}}
\newcommand{\ens}[1]{\forcemath{{\ensemble_{\mathsf{#1}}}}}

\newcommand{\neigh}[1]{\forcemath{\mathcal{N}\left(#1\right)}}
\newcommand{\vn}{\forcemath{\mathsf{v}}}
\newcommand{\cn}{\forcemath{\mathsf{c}}}
\newcommand{\msg}[2]{\forcemath{m_{#1\rightarrow#2}}}
\newcommand{\mch}{\forcemath{m_{\mathrm{ch}}}}
\newcommand{\sign}{\forcemath{\mathrm{sign}}}

\newcommand{\vnt}{\forcemath{\mathsf{V}}}
\newcommand{\cnt}{\forcemath{\mathsf{C}}}

\newcommand{\degCN}{\ensuremath{d_c}}

\newcommand{\weight}[1]{\mathsf{wht}\left(#1 \right)}

\newcommand{\field}{\mathbb{F}}

\renewcommand{\vec}[1]{\ensuremath{\bm{#1}}}

\newcommand{\WF}{\mathrm{WF}}
\newcommand{\WFisd}{\WF_{\mathsf{ISD}}}
\newcommand{\WFdist}{\WF_{\mathsf{dist}}}
\newcommand{\WFdec}{\WF_{\mathsf{dec}}}

\newcommand{\mat}[1]{\ensuremath{\bm{#1}}}
\newcommand{\F}{\ensuremath{\mathbb{F}}}
\renewcommand{\vec}[1]{\ensuremath{\bm{#1}}}
\newcommand{\Dec}[1]{\ensuremath{\text{DEC}_{#1}}}

\begin{document}
	
\begin{acronym}
\acro{LDPC}{low-density parity-check}
\acro{MDPC}{Moderate-density parity-check}
\acro{VN}{variable node}
\acro{CN}{check node}
\acro{BP}{belief propagation}
\acro{EXIT}{extrinsic information transfer}
\acro{APP}{a posteriori probability}
%%%%%%
\acro{MI}{mutual information}
%%%%%%
\acro{FER}{frame error rate}
\acro{QC}{quasi cyclic}
\acro{SPA}{sum-product algorithm}
\acro{DE}{density evolution}
\acro{BSC}{binary symmetric channel}
\acro{MMT}{May, Meurer and Thomae}
\acro{ISD}{information set decoding}
\acro{MET}{multi-edge type}
\acro{DOOM}{decode-one-out-of-many}
\end{acronym}	
	
\title{Protograph-based Quasi-Cyclic MDPC Codes for McEliece Cryptosystems}
\author{Gianluigi Liva and Hannes Bartz
\thanks{G. Liva and H. Bartz are with with the Institute of Communication and Navigation of the Deutsches Zentrum f\"{u}r Luft- und Raumfahrt (DLR), 82234 Wessling,
	   Germany (e-mail: {\tt \{gianluigi.liva,hannes.bartz\}@dlr.de})}}

\thispagestyle{empty} \setcounter{page}{1} 
\pagenumbering{gobble}

\maketitle

\begin{abstract}
	In this paper, ensembles of quasi-cyclic  moderate-density parity-check (MDPC) codes based on protographs are introduced and analyzed in the context of a McEliece-like cryptosystem. The proposed  ensembles significantly improve the error correction capability of the regular MDPC code ensembles that are currently considered for post-quantum cryptosystems without increasing the public key size. The proposed ensembles are analyzed in the asymptotic setting via density evolution, both under the sum-product algorithm and a low-complexity (error-and-erasure) message passing algorithm. The asymptotic analysis is complemented at finite block lengths by Monte Carlo simulations.
	 The enhanced error correction capability remarkably improves the scheme robustness with respect to (known) decoding attacks.
\end{abstract}


\section{Introduction}\label{sec:intro}

{\ac{MDPC} codes \cite{ouzan2009MDPC} have been recently proposed as underlying coding scheme for McEliece-like cryptosystems \cite{Misoczki13:MDPC,baldi2014qc}. The family of \ac{MDPC} codes admit a parity-check matrix of moderate density\footnote{The existence of a \ac{MDPC} matrix for a binary linear block code does not rule out the possibility that the same code fulfills a (much) sparser parity-check matrix. As in most of the literature, we neglect the probability that a code defined by a randomly-drawn moderate parity check matrix admits a sparser parity-check matrix.}, yielding codes  with large minimum distance.
In~\cite{Misoczki13:MDPC}{,} a McEliece cryptosystem based on \ac{QC}-\ac{MDPC} codes that defeats information set decoding attacks on the dual code due to the moderate density parity-check matrix{,} is presented. For a given security level, the \ac{QC}-\ac{MDPC} cryptosystem allows for very small key sizes compared to other McEliece variants.} 

{In this paper, we introduce a variation on the scheme of~\cite{Misoczki13:MDPC}. In particular, we investigate the adoption of \ac{MDPC} ensembles based on protographs \cite{Thorpe03:proto} to improve the error correction capability of the code underlying the cryptosystem. We focus on protographs containing state \acp{VN}. The use of state \acp{VN} allows designing codes that can be decoded over an extended Tanner graph that is sparser than the reduced Tanner graph associated with the code{'s} (moderate-density) parity-check matrix.} 

{For the introduced ensembles, we relate the density of the extended Tanner graph to the density of the reduced Tanner graph. The lower density of the extended Tanner graph allows remarkable gains in terms of error correction capability with respect to the ensembles of~\cite{Misoczki13:MDPC}. The improvement is analyzed both in the asymptotic setting via \ac{DE} \cite{richardson01:capacity} and at finite block length via Monte Carlo simulations, under two decoding algorithms: The \ac{SPA} and the less complex Algorithm E from  \cite{richardson01:capacity}. For one of the proposed ensembles, a gain of $40\%$ in the weight of error patterns (decodable with a target error probability) is observed with respect to the ensembles of \cite{Misoczki13:MDPC}, rendering the scheme more resilient to \ac{ISD} attacks \cite{stern1988method}.}
\section{Preliminaries}\label{sec:prelim}

\subsection{Circulant Matrices}

A binary circulant matrix $\bm{A}$ of size $Q$ is a $Q\times Q$ matrix with coefficients in $\field_2$ obtained by cyclically shifting its first row $\bm{a}=\left(a_0, a_1, \ldots, a_{Q-1}\right)$ to the right, yielding
\begin{equation}
	\bm{A}=\left(\begin{array}{cccc}
		a_0 & a_1 & \cdots & a_{Q-1}\\
		a_{Q-1} & a_0 & \cdots & a_{Q-2}\\
		\vdots & \vdots & \ddots &\vdots\\
		a_1 & a_2 & \cdots & a_{0}\\
	\end{array}\right).
\end{equation} 
The set of $Q\times Q$ circulant matrices together with the matrix multiplication and addition forms a commutative ring and it is isomorphic to the polynomial ring $\left(\field_2[X]/\left(X^Q-1\right),+,\cdot\right)$. 
In particular, we can associate to  the circulant $\bm{A}$ a polynomial $a(X)=a_0+a_1X+\ldots+a_{Q-1}x^{Q-1} \in \field_2[X]$. 
Consider two $Q\times Q$ circulants $\bm{A}$ and $\bm{B}$ and their associated polynomials $a(X)$ and $b(X)$. 
Denote by $\bm{C}=\bm{A}+\bm{B}$ and $\bm{D}=\bm{A}\bm{B}$. 
Then, the circulants $\bm{C}$ and $\bm{D}$ are associated to the polynomials $c(X)=a(X)+b(X)$ and $d(X)=a(X)\cdot b(X) \mod{\left(X^Q-1\right)}$. We indicate the vector of coefficients of a polynomial $a(X)$ as $\bm{a}=\left(a_0, a_1, \ldots, a_{Q-1}\right)$. The weight of a polynomial $a(X)$ is the number of its non-zero coefficients.
We indicate both weights with the operator $\weight{\cdot}$, i.e., $\weight{a(X)}=\weight{\bm{a}}$.

\subsection{QC-MDPC-based Cryptosystems}

A binary \ac{MDPC} code of length $n$, dimension $k$ and row weight $\degCN$ is defined by a binary parity-check matrix $\mat{H}$ whose rows have (moderate) Hamming weight $\degCN$. For $n=NQ$, dimension $k=KQ$, redundancy\footnote{{We assume here parity-check matrices without redundant rows.}} $m=MQ$ with $M=N-K$ for some integer $Q$, the parity-check matrix $\mat{H}(X)$ of a \ac{QC}-\ac{MDPC}
 code in polynomial form is a $M\times N$ matrix. Without loss of generality we consider in the following codes with $M=1$. 
This family of codes covers a wide range of code rates and is of particular interest for cryptographic applications since the parity-check and generator matrices can be described in a very compact way.
The parity-check matrix of \ac{QC}-\ac{MDPC} codes with $M=1$ has the form 
\begin{equation}\label{eq:H_QC-MDPC-simple}
\mat{H}(X)=
\begin{pmatrix}
h_{0}(X) & h_{1}(X) & \dots & h_{N-1}(X)
\end{pmatrix}.
\end{equation}
Let $\Dec{\mat{H}}(\cdot)$ be an efficient decoder for the code defined by the parity-check matrix $\mat{H}$. The cryptosystem operates in the following manner.

\subsubsection{Key Generation} The private key is generated as a parity-check matrix  of the form \eqref{eq:H_QC-MDPC-simple} with $\weight{h_{i}(X)}=\degCN^{(i)}$ for $i=0,\dots,N$.
	The matrix $\mat{H}$ with row weight $\degCN=\sum_{i=0}^{N-1}\degCN^{(i)}$ is the \emph{private} key. The \emph{public} key is the corresponding binary $k\times n$ generator matrix $\mat{G}$ in systematic form\footnote{{In \cite{Misoczki13:MDPC} the use of a CCA2-secure conversion is proposed, which allows using  $\mat{G}$ in systematic-form without leaking information.}} (note that the generator matrix can be described by $KQ$ bits, yielding a small public key size).

\subsubsection{Encryption} 
	To encrypt a plaintext $\vec{u}\in\F_2^{k}$ a user computes the {cyphertext} $\vec{c}\in\F_2^n$ using the public key $\mat{G}$ as 
	${\vec{c}=\vec{x}+\vec{e}}$
	{where $\vec{x}=\vec{u}\mat{G}$ and $\vec{e}$} is an error vector uniformly chosen from all vectors from $\F_2^{n}$ of Hamming weight $\weight{\bm{e}}= e$ .

\subsubsection{Decryption} 
 To decrypt a {cyphertext} the authorized recipient uses the private key  to obtain ${\hat{\vec{x}}=\Dec{\mat{H}}(\vec{c})}$.
{Since $\mat{G}$ is in systematic form the plaintext corresponds (in case of correct decoding) to the first $k$ bits of $\hat{\vec{x}}$.}

\subsection{Protograph-based Codes}

A protograph $\proto$ \cite{Thorpe03:proto} is a small bipartite graph comprising a set of $N_0$ \acp{VN} (also referred to as \ac{VN} types) $\left\{\vnt_0, \vnt_1, \ldots, \vnt_{N_0-1}\right\}$ and a set of $M_0$ \acp{CN} (i.e., \ac{CN} types)  $\left\{\cnt_0, \cnt_1, \ldots, \cnt_{M_0-1}\right\}$. A \ac{VN} type $V_j$ is connected to a \ac{CN} type $\cnt_i$ by $b_{ij}$ edges. A protograph can be equivalently represented in matrix form by an $M_0\times N_0$ matrix $\B$. The $j$th column of $\B$ is associated to  \ac{VN} type $\vnt_j$ and the $i$th row of $\B$ is associated to \ac{CN} type $\cnt_i$. The $(i,j)$ element of $\B$ is  $b_{ij}$. 
A larger graph (derived graph) can be obtained from a protograph by applying a copy-and-permute procedure. The protograph is copied $Q$ times, and the edges of the different  copies are permuted preserving the original protograph connectivity: If a type-$j$ \ac{VN} is connected to a type-$i$ \ac{CN} with $b_{ij}$ edges in the protograph, in the derived graph each type-$j$ \ac{VN} is connected to $b_{ij}$ distinct type-$i$ \acp{CN} (observe that multiple connections between a \ac{VN} and a \ac{CN} are not allowed in the derived graph). The derived graph is a Tanner graph with $n_0=N_0Q$ \acp{VN} and $m_0=M_0Q$ \acp{CN} that can be used to represent a binary linear block code.
A protograph $\proto$ defines a code ensemble $\ensemble$. For a given protograph $\proto$, consider all its possible derived graphs with $n_0=N_0Q$ \acp{VN}. The ensemble $\ensemble$ is the collection of codes associated to the derived graphs in the set.

Protographs allow specifying graphs which contain \acp{VN} which are associated to codeword symbols, as well as \acp{VN} which are not associated to codeword symbols. The latest class of \acp{VN} are often referred to as \emph{state} or \emph{punctured} \acp{VN}. The term ``punctured'' is used since the code associated with the derived graph can be seen as a punctured version of a longer code associated with the same graph for which all the \acp{VN} are associated to codeword bits. The introduction of state \acp{VN} in a code graph allows designing codes with a remarkable performance \cite{Divsalar09:ProtoJSAC}. 
The \ac{QC}{-}\ac{MDPC} codes introduced in \cite{Misoczki13:MDPC} admit a protograph representation. In particular, for the rate $1/2$ case, the base matrix has the form
\begin{align}
	\bm{B}=\left(
	\begin{array}{cc}
		b_{00} & b_{01} \label{eq:ensA}
	\end{array}
	\right). 
\end{align}
Here, $N_0=N=2$ and  $M_0=M=1$.
The expansion of the protograph can be performed in a structured manner yielding a \ac{QC}{-}\ac{MDPC} with  parity-check matrix in polynomial form
$\bm{H}(X)=\left(
	\begin{array}{cc}
		h_{00}(X) & h_{01}(X) \\
	\end{array}
	\right)$
with $\weight{h_{ij}(X)}=b_{ij}$.
In the remainder of the paper, we will focus on rate $1/2$ codes, and we will denote the ensemble defined by \eqref{eq:ensA} as the \emph{reference} ensemble or $\ens{A}$. In particular, we will consider as a study case the shortest public key size considered in \cite{Misoczki13:MDPC} {for $80$ bit security} where $Q=4801$ is the public key size, and $n=9602$ is the \ac{MDPC} code block length. For this specific case, the reference ensemble is obtained by setting $b_{00}=b_{01}=45$ {yielding a total row weight of $90$. This choice of parameters was introduced in \cite{Misoczki13:MDPC} to obtain the security level of $80$ bits.}

\subsection{Decoding Algorithms}

Denote by $\vn_0,\vn_1,\ldots,\vn_{n_0-1}$ the $n_0$ \acp{VN} in the code Tanner graph, and by  $\cn_0,\cn_1,\ldots,\cn_{m_0-1}$ then $m_0$ \acp{CN}. The neighborhood of a \ac{VN} \vn~is \neigh{\vn}, and similarly \neigh{\cn} denotes the neighborhood of the \ac{CN} \cn. The message from~\vn~to~\cn~is denoted by $\msg{\vn}{\cn}$, and the message from~\cn~to~\vn~is denoted by $\msg{\cn}{\vn}$. The channel message is $\mch$. If a codeword bit is associated to~\vn~(i.e.,~\vn~is not a state \ac{VN}), then we denote it with {a slight abuse of notation} by $x$, with $x=+1$ stands for a $0$ and $x=-1$ stands for a $1$. The bit in the cyphertext associated to $x$ is $c$, with $c=-x$ if an error is introduced, $c=x$ otherwise. For state \acp{VN}, no cyphertext bits are produced at the encryption stage. We consider next two decoding algorithms. The first algorithm is the classical \ac{SPA}, for which we introduce a generalization allowing an attenuation of the extrinsic information produced at the \acp{CN}. As we shall see, the attenuation can be used as a heuristic method to improve the performance at low error rates. The second algorithm is the Algorithm E introduced in \cite{richardson01:capacity}, which reduces the decoding complexity by limiting the message alphabet to the ternary set $\{-1,0,+1\}$. As observed in \cite{richardson01:capacity}, Algorithm E benefits also from the introduction of a heuristic scaling parameter. Here, nevertheless, the parameter is used to amplify the channel message, and it plays a role not only at low error rates, but also in the so-called waterfall performance of the code. While in \cite{richardson01:capacity} it was suggested to vary the scaling parameter with the iteration number, here we will keep the scaling parameter fixed throughout the iterations. In both cases, the scaling parameter is indicated by $\omega$.

\subsubsection{Scaled Sum-Product Algorithm}
The channel message is initialized as 
\[
\mch=c \ln \frac{n-e}{n}
\]
for the \acp{VN} associated to a cyphertext bit (recall that $c=\pm1$ is the cyphertext bit), whereas $\mch=0$ for the state \acp{VN}.
At the \acp{VN},
\begin{equation}
	\msg{\vn}{\cn}=\mch+\sum_{\cn'\in\neigh{\vn}\backslash\cn}{\msg{\cn'}{\vn}} \label{eq:SPA_VN}
\end{equation}
while at the \acp{CN}
\begin{equation}
	\msg{\cn}{\vn}=\omega 2\tanh^{-1}\left[ \prod_{\vn'\in\neigh{\cn}\backslash\vn}\tanh\left(\frac{\msg{\vn'}{\cn}}{2}\right)\right] \label{eq:SPA_CN}
\end{equation}
with final decision, after iterating \eqref{eq:SPA_VN}, \eqref{eq:SPA_CN} a given number of times, given by 
\begin{equation}
	\hat{x}=\sign\left[\mch+\sum_{\cn\in\neigh{\vn}}{\msg{\cn}{\vn}}\right] \label{eq:SPA_final}.
\end{equation}

\subsubsection{Algorithm E}
The channel message is initialized as $\mch=c$
for the \acp{VN} associated to a cyphertext bit (recall that $c=\pm1$ is the cyphertext bit), whereas $\mch=0$ for the state \acp{VN}.
At the \acp{VN},
\begin{equation}
	\msg{\vn}{\cn}=\sign\left[\omega \mch+\sum_{\cn'\in\neigh{\vn}\backslash\cn}{\msg{\cn'}{\vn}}\right] \label{eq:ALGE_VN}
\end{equation}
while at the \acp{CN}
\begin{equation}
	\msg{\cn}{\vn}= \prod_{\vn'\in\neigh{\cn}\backslash\vn}\msg{\vn'}{\cn} \label{eq:ALGE_CN}
\end{equation}
with final decision, after iterating \eqref{eq:ALGE_VN}, \eqref{eq:ALGE_CN} a given number of times, given by 
\begin{equation}
	\hat{x}=\sign\left[\mch+\sum_{\cn\in\neigh{\vn}}{\msg{\cn}{\vn}}\right] \label{eq:ALGE_final}.
\end{equation}
{Observe that the scaling parameter is de-activated (i.e., it is set to $1$) in the final decision.}
\section{Protograph-based MDPC Ensembles}\label{sec:de}

We introduce next protograph-based \ac{MDPC} ensembles which make use of state \acp{VN}. For all the ensembles we adopt a base matrix in the form
\begin{align}
	\bm{B}=\left(
	\begin{array}{c|cc}
		1 & b_{01} & b_{02} \\
		b_{10} & b_{11} & b_{12} \\
	\end{array}
	\right) \label{eq:otherens}
\end{align}
where the column on the left of the vertical delimiter is associated to state \acp{VN}. The expansion of the protograph can be performed in a structured manner yielding a \ac{QC}{-}\ac{MDPC}. In particular, we first obtain a binary matrix in polynomial form
\begin{align}
	\bm{\Gamma}(X)=\left(
	\begin{array}{c|cc}
		1 & \gamma_{01}(X) & \gamma_{02}(X) \\
		\gamma_{10}(X) & \gamma_{11}(X) & \gamma_{12}(X) \\
	\end{array}
	\right) \label{eq:gamma}
\end{align}
with $\weight{\gamma_{ij}(X)}=b_{ij}$. The parity-check matrix  is
\begin{align}
	\bm{H}(X)=\left(
	\begin{array}{cc}
		h_{00}(X) & h_{01}(X) \\
	\end{array}
	\right) \label{eq:H}
\end{align}
with 
$h_{00}(X)=\left[\gamma_{11}(X)+\gamma_{01}(X)\gamma_{10}(X)\right]{\mod (X^Q-1)}$
and $h_{01}(X)=\left[\gamma_{12}(X)+\gamma_{02}(X)\gamma_{10}(X)\right]{\mod (X^Q-1)}$.
In this case, $N_0=3$ and $M_0=2$, yielding a Tanner graph (associated with the matrix $\bm{\Gamma}$) with $n_0=3Q$ \acp{VN}. By comparing with the block length $n=2Q$, one finds that the code defined by \eqref{eq:H} can be described as punctured version of a longer code which has $\bm{\Gamma}$ as {its} parity-check matrix, where the first $Q$ bits in each codeword are punctured.
The following proposition established a relationship between the weights of the rows of \eqref{eq:gamma} and the weight of the rows of \eqref{eq:H}.
\begin{prop}\label{prop:weight}
	The Hamming weight of each row of $\bm{H}$ is upper bounded by $b_{11}+b_{01}b_{10}+b_{12}+b_{02}b_{10}$.
\end{prop}
\begin{proof}
	Follows by  triangle's inequality.	
\end{proof}
In the remainder of the paper, coefficients of the base matrix are chosen such that upper bound given  by Proposition~\ref{prop:weight} matches the weight of the rows of the parity check matrix of the reference ensemble, i.e., $b_{01}b_{10}+b_{11}+b_{02}b_{10}+b_{12}=90$.

The number of distinct matrices~\eqref{eq:gamma} with a base matrix of the form~\eqref{eq:otherens}, i.e. the number of different private keys in an ensemble $\ens{}$, is
\begin{align}\label{eq:sizeKeySpace}
 {\mathsf{N}}(\ens{})=\frac{1}{Q}\binom{Q}{1}\prod_{i,j}\binom{Q}{b_{ij}}=\prod_{i,j}\binom{Q}{b_{ij}}.
\end{align} 
The size of the key space ${\mathsf{N}}(\ens{})$ for different ensembles is given in Table~\ref{tab:thresholds}.
{We provide next a \ac{DE} analysis of protograph-based ensembles.}\footnote{{The codes that will be adopted in practice are \ac{QC}, and hence represent a sub-ensemble of the (larger) protograph-based \ac{MDPC} ensemble. Anyhow, we consider the results of the analysis as accurate in predicting the error correction capability of the corresponding protograph-based \ac{QC}-\ac{MDPC} ensembles.}}

\medskip

\subsubsection{Sum-Product Algorithm}

When the \ac{SPA} is used, we resort to quantized \ac{DE}.  We refer to  \cite{Chung01:DE} for the details. The extension to protograph ensembles is straightforward and follows the footsteps of \cite{liva_pexit,Liva2013:Proto_TWCOM}. Simplified approaches based on the Gaussian approximation are discarded due to the large \ac{CN} degrees \cite{chung2001analysis} used by the \ac{MDPC} code ensembles.

\medskip

\subsubsection{Algorithm E}

In the following, we provide an extension of the \ac{DE} analysis of protograph ensembles for Algorithm E, which was originally introduced in \cite{richardson01:capacity} for unstructured ensembles. Rather than stating the complete \ac{DE}, we sketch the analysis by showing how message probabilities are updated at \acp{VN} and \acp{CN}.
Let us consider the transmission over a \ac{BSC} with error probability $\delta$. We make use of the conventional assumption of the all-zero codeword transmission. It follows that all the message probabilities derived next are conditioned to the transmission of a zero value. Due to the mapping $(0 \leftrightarrow +1, \, 1 \leftrightarrow -1)$ at the decoder input, we have that messages exchanged by \acp{VN} and \acp{CN} take values in $\{-1,0,+1\}$, whereas the messages associated with the channel observations take values in $\{-\omega,0,+\omega\}$. In either case, $0$ is the value associated with an erasure (which replaces the channel observation in the case of a state \ac{VN}).
We consider first a degree-$d$ \ac{CN} $\cn$ of a given type. We assume the \ac{CN} to be connected to $d$ \acp{VN}, each of different type (having some~\acp{VN} of the same type being a particular case). For ease of notation, we assume the $d$ \acp{VN} to be $\vn_0, \vn_1, \ldots, \vn_{d-1}$. At a given iteration, we denote by 
$q_{\ell}^{(i)}$ the probability that the message from $\vn_{i}$ to $\cn$ takes value $\ell$. Similarly, we denote by 
$p_{\ell}^{(i)}$ the probability that the message from  $\cn$ to  $\vn_{i}$ takes value $\ell$, after the \ac{CN} elaboration. We have that
\begin{align}
p_{-1}^{(i)}&=\frac{1}{2}\left[\prod_{j\neq i}\left(q_{+1}^{(j)}+q_{-1}^{(j)} \right) - \prod_{j\neq i}\left(q_{+1}^{(j)}-q_{-1}^{(j)} \right)    \right] \label{eq:DEE_CN_1}\\
p_{0}^{(i)}&=1-\prod_{j\neq i}\left(1- q_0^{(j)}   \right)
\label{eq:DEE_CN_2}\\
p_{+1}^{(i)}&=\frac{1}{2}\left[\prod_{j\neq i}\left(q_{+1}^{(j)}+q_{-1}^{(j)} \right) + \prod_{j\neq i}\left(q_{+1}^{(j)}-q_{-1}^{(j)} \right)    \right].  \label{eq:DEE_CN_3}
\end{align}
It follows that \eqref{eq:DEE_CN_1}, \eqref{eq:DEE_CN_2} and \eqref{eq:DEE_CN_3} fully describe the evolution of the message probabilities at the \acp{CN}. Let us consider now a degree-$d$  \ac{VN} $\vn$ of a given type. We assume the \ac{VN} to be connected to $d$ \acp{CN}, each of different type (having some \acp{CN} of the same type being a particular case). Again, for ease of notation, we assume the $d$ \acp{CN} to be $\cn_0, \cn_1, \ldots, \cn_{d-1}$.
We shall introduce next the  probability vectors 
\begin{align}
\bm{p}^{(i)}:=\left(p_{-1}^{(i)}, p_{0}^{(i)}, p_{+1}^{(i)}  \right) \quad \bm{q}^{(i)}:=\left(q_{-1}^{(i)}, q_{0}^{(i)}, q_{+1}^{(i)}  \right).
\end{align}
{Moreover, we introduce the channel message probability vector $\bm{m}:=\left(m_{-\omega},m_{-\omega+1},\dots,m_{+\omega}  \right)$}
with $m_0=1$ and $m_\ell=0$ for all $\ell\neq 0$ if $\vn$ is a state \ac{VN}, whereas $m_{-\omega}=\delta$, $m_{+\omega}=1-\delta$ and  $m_\ell=0$ for all $\ell\neq \pm \omega$ otherwise. We introduce the intermediate probability vector
$\bm{z}^{(i)}:= \left[ \bigoasterisk_{j\neq i}  \bm{p}^{(j)}  \right] \circledast \bm{m}$
i.e., $\bm{z}^{(i)}$ is the convolution of the probability vectors associated to the messages at the input of $\vn$, with the exception of the one received from $\cn_i$. We have that
\begin{align}
q_{-1}^{(i)}=\sum_{\ell<0}{z^{(i)}_\ell} \qquad
q_{0}^{(i)}=z^{(i)}_0  \qquad
q_{+1}^{(i)}=\sum_{\ell>0}{z^{(i)}_\ell}. \label{eq:DEE_VN_3}
\end{align}
Observe that  \eqref{eq:DEE_VN_3} fully describe the evolution of the message probabilities at the \acp{CN}. By iterating them with \eqref{eq:DEE_CN_1}, \eqref{eq:DEE_CN_2} and \eqref{eq:DEE_CN_3} for all \ac{VN}/\ac{CN} types specified by the protograph, one can track the evolution of the message probabilities.  
Denote by $\bm{f}=\left(f_{-1},f_0,f_{+1}\right)$ the final \ac{APP} estimate probability vector at $\vn$ (obtained after a given number of iterations). We have that 
\begin{align}
f_{-1}=\sum_{\ell<0}{z_\ell} \qquad f_{0}=z_0 \qquad f_{+1}=\sum_{\ell>0}{z_\ell}
\end{align}
with
$
\bm{z}:= \left[ \bigoasterisk_{j}  \bm{p}^{(j)}  \right] \circledast \bm{m}'
$
{and $\bm{m}':=(m'_{-1},m'_{0},m'_{+1})$ and $m'_{-1}=\delta$, $m'_{0}=0$ and $m'_{+1}=1-\delta$ if $\vn$ is not a state \ac{VN}, and $m'_{-1}=0$, $m'_{0}=1$ and $m'_{+1}=0$ otherwise.}
Successful decoding convergence is declared if $f_{-1}\rightarrow 0$ and $f_{0}\rightarrow 0$ as the iteration count tends to infinity, for all \ac{VN} types.

\medskip

\subsubsection{Asymptotic Error Correction Capability Estimates}

{To estimate the error correction capability of the ensembles under consideration, we computed the iterative decoding threshold over a \ac{BSC} with error probability $\delta$. The thresholds are computed under the \ac{SPA} and Algorithm E decoding. We denote the former as $\thrspa$ and the latter as $\thre$. Table \ref{tab:thresholds} provides the decoding threshold for some ensembles with base matrix in the form \eqref{eq:otherens}, together with the reference ensemble from  \cite{Misoczki13:MDPC}. In particular, we provide the products $n\thrspa$ and $n\thre$ as rough estimates of the error correction capability of codes drawn from ensembles with finite block length $n$, for the study case of $n=9602$. The values represent a first estimate of the maximum error pattern weights for which decoding succeeds with high probability. The ensemble that exhibits the largest iterative decoding threshold under both \ac{SPA} and Algorithm E decoding is  ensemble $\ens{C}$. The gain in the weight of error patterns (decodable with high probability) with respect to the reference ensemble is here in the order of $50\%$ under the \ac{SPA}, while under Algorithm E the gain reduces to $20\%$.}

\renewcommand{\arraystretch}{1.4}
\newcommand\hh{\noalign{\vskip 1mm} \hline \noalign{\vskip 1mm}} 
\newcommand{\mr}[2]{\multirow{ #1}{*}{#2}}
\begin{table}[t]
	\begin{center}
		\caption{Thresholds computed for different protographs, $n=9602$}\label{tab:thresholds}
		\begin{tabular}{ccccc}
			\hline\hline
			Ensemble & Base matrix & ${\mathsf{N}}(\ens{})$ & $n\thrspa\, (\omega)$ & $n\thre\, (\omega)$ \\ \hline
			\mr{2}{\ens{A}}& \mr{2}{\BMA{45}{45}}  & \mr{2}{$2^{715}$} & $113\,(1)$ & $57\,(1)$ \\
			&			   &       &     $112\,(0.5)$       & $106\,(14)$  \\ \hh  						   
			\mr{2}{\ens{B}}&\mr{2}{\BMB{1}{5}{8}{5}{8}{5}} & \mr{2}{$2^{328}$} & \mr{2}{$132$}& $25\,(1) $ \\			   
			&				&   			 &          & $57\,(4) $ \\		\hh	
			\mr{2}{\ens{C}}&\mr{2}{\BMB{1}{2}{22}{1}{22}{1}} & \mr{2}{$2^{446}$} & $171\,(1)$ & $ 43 \,(1) $ \\			   
			&				 &  			 &      $155\,(0.8)$     & $128\,(8) $ \\	\hh   			   			
		\end{tabular}\vspace{-0.3cm}
	\end{center}
\end{table}

\subsection{Error Correction Capability at Finite Length}

{Monte Carlo simulations have been performed to measure the actual error correction capability of the codes with base matrices in the form \eqref{eq:otherens}. In particular,  for a given ensemble and the given block length $n=9602$, \ac{QC}{-}\ac{MDPC} codes have been obtained by expanding the base matrix with $4801 \times 4801$ circulant matrices of suitable weight, with the circulants picked uniformly at random. The results, in terms of block error rate vs. weight of the error pattern, for the reference ensemble and for the ensemble $\ens{C}$ are depicted in Figure \ref{fig:9602SPA100it}. The performance measured with Algorithm E are in good accordance with the gain predicted by the iterative decoding threshold analysis. At a block error rate of $10^{-5}$ the code from $\ens{C}$ allows operating with around $20$ errors more than what is allowed by the code drawn from the reference ensemble.} 

{Under \ac{SPA} decoding, the codes drawn from both ensembles showed performance curves with signs of slope change at moderate-high error floors, preventing the achievement of low block error rates with reasonably-high error pattern weights.  We conjecture that the reason for this behavior might be found in the numerous trapping sets affecting the dense graphs (recall that the code construction does not leverage on any girth optimization technique, due to the need of generating the private key uniformly at random). To mitigate this effect, we made use of the attenuation parameter $\omega$ in \eqref{eq:SPA_CN} to reduce the magnitude of the extrinsic estimated provided at the output of the \acp{CN}.\footnote{{The use of soft information scaling was used in a similar manner in \cite{Pyndiah98} to improve the performance of block turbo codes in the error floor region.}} This heuristic approach turns out to be effective in improving the error correction capability at low block error rates. The choice of the scaling factors has been carried out by searching via simulations the largest value of $\omega$ for which no sign of error floor appears at a block error rate greater than $10^{-5}$. Surprisingly, the use of a scaling parameter $\omega=0.5$ with codes from the reference ensemble allows attaining remarkable gain at low error rates without sacrificing the waterfall region performance (the later result in accordance with the \ac{DE} analysis). For the code designed from ensemble $\ens{C}$, the introduction of the scaling coefficient entails a visible loss in the waterfall region (again, in accordance with the \ac{DE} analysis). We expect that a further optimization of the algorithm (e.g., by allowing variable scaling coefficients across iterations/edge types) may reduce the loss in the waterfall region. {Even accounting} for the loss introduced by the scaling coefficient, codes drawn from ensemble $\ens{C}$ show gains in error correction capability in order of $40\%$ {compared to \ens{A}}.}

\begin{figure}[t]
	\centering
	\includegraphics[width=0.93\columnwidth]{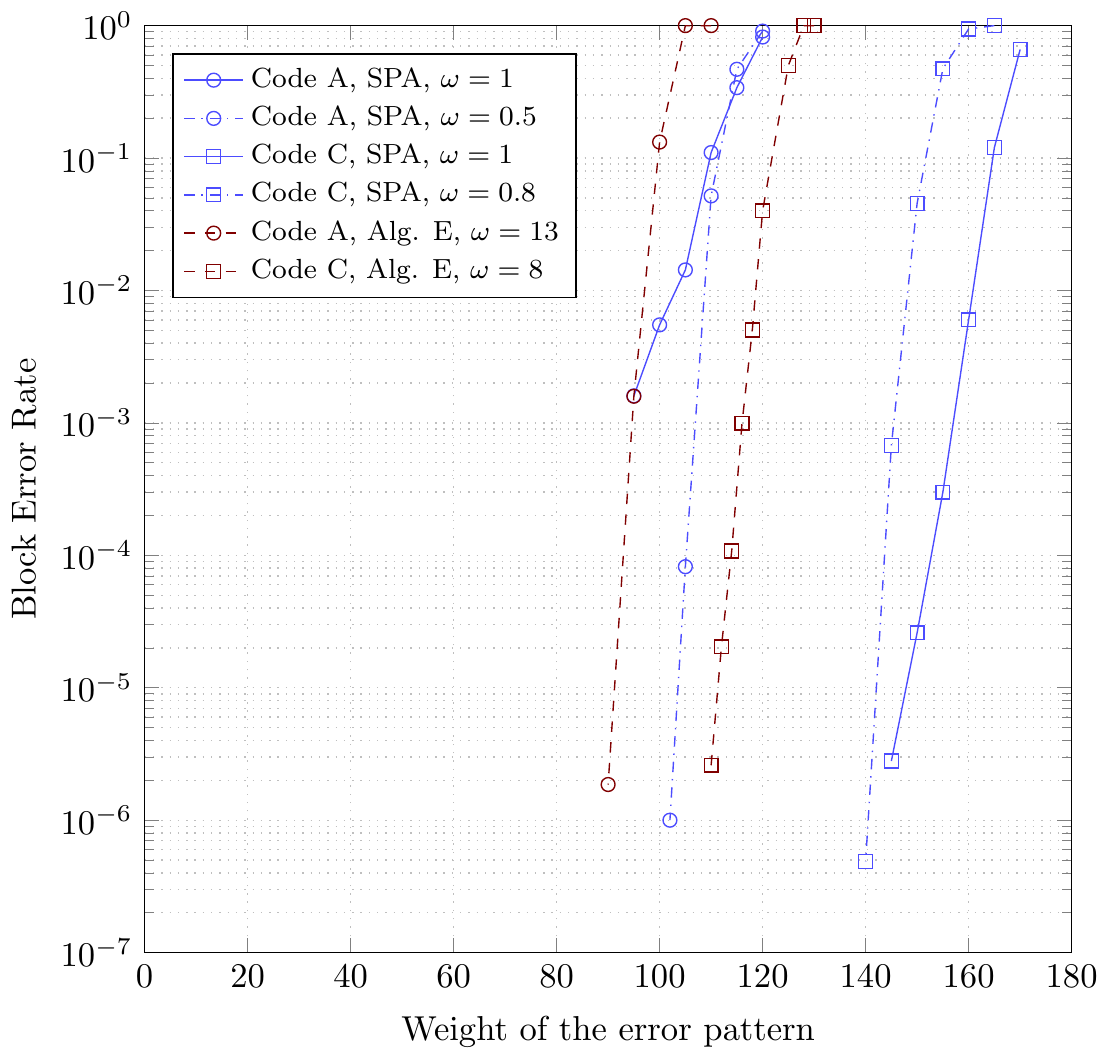}
	\caption{{Block error rate for some codes from the ensembles $\ens{A}$ and $\ens{C}$ in Table~\ref{tab:thresholds}. \ac{SPA} (with and without scaling) and Algorithm E with $100$ iterations.}}
	\label{fig:9602SPA100it}
\end{figure}

\section{Security}\label{sec:security}

In this section we estimate the security level of the proposed protograph-based \ac{MDPC} McEliece cryptosystem. 
For the analysis we use the pessimistic assumption that the system is broken as soon as the \ac{QC}-\ac{MDPC} parity-check matrix $\bm{H}(X)$ in~\eqref{eq:H} is reconstructed. 
This means that we assume that the effort for obtaining the extended  matrix $\bm{\Gamma}(X)$ (see~\eqref{eq:gamma}) from $\bm{H}(X)$ is below the security level of the cryptosystem.

{We denote by $\WFisd(n,m,e)$ the cost of decoding an error pattern of Hamming weight $e$ with a  $(n,m)$ linear code (\ac{ISD} is assumed here).}
The cost of distinguishing a key, i.e. to recover one {weight-$\degCN$} row of the (sparse) parity check matrix $\bm{H}(X)$ is denoted by $\WFdist(n,m,\degCN)$. %C{$w$ undefined - shall it be $\degCN$?}
For the \ac{QC} case the cost of recovering the whole secret key equals $\WFdist(n,m,\degCN)$.
We compute the work factors of the key distinguishing attack $\WFdist(n,m,\degCN)$ and the work factor for the key recovery attack and the decoding attack $\WFdec(n,m,\degCN)$ for the \ac{QC}-\ac{MDPC} McEliece cryptosystem according to~\cite[Tab.~1]{Misoczki13:MDPC}, i.e.,
$\WFdist(n,m,\degCN)=\WFisd(n,n-m,\degCN)/m$ and $\WFdec(n,m,e)=\WFisd(n,m,e)/{\sqrt{m}}$.
The work factor estimates include the possible gains obtained by using the decode-one-out-of-many approach~\cite{sendrier2011-DOOM}.
We use the non-asymptotic results from~\cite[Sec.~3.3]{hamdaoui2013non} to estimate the  work factor for the May-Meurer-Thomae variant of \ac{ISD}~\cite{may2011decoding}.
Consider the \ac{QC}-\ac{MDPC} code ensemble $\ens{A}$ for $80$ bit security  with $n=9602, k=4801, \degCN=90$ from~\cite{Misoczki13:MDPC}.
For decoding we consider Algorithm E.
The work factor for the key distinguishing attack for $\ens{A}$ and $\ens{C}$ is $\WFdist(9602,4801,90)=2^{80.6}$. 
{Figure~\ref{fig:9602SPA100it} shows that for $\ens{C}$ we have a block error rate of $10^{-6}$ for $e\approx102$ whereas for $\ens{A}$ we have $e\approx 84$.
It follows that the work factor of the decoding attack for ensemble is, according to~\cite[Tab.~1]{Misoczki13:MDPC}, $2^{81.0}$ for the reference ensemble, and $2^{98.3}$ for the scheme based on the ensemble $\ens{C}$.}

\section{Conclusions}\label{sec:conclusions}
Protograph-based moderate-density parity-check (MDPC) code ensembles are introduced and analyzed in the context of a McEliece-like cryptosystem. The proposed  ensembles significantly improve the error correction capability of the MDPC code ensembles that are currently considered for post-quantum cryptosystems, without increasing the public key size.
	The enhanced error correction capability remarkably improves the robustness with respect to decoding attacks.


\end{document}